\documentclass[10pt,a4paper]{amsart}
\usepackage[margin=20mm]{geometry}

\usepackage{amssymb,amsmath}
\usepackage{amsthm}
\usepackage{bbm}
\usepackage{stmaryrd}
\usepackage{hyperref}
\usepackage[usenames,dvipsnames]{xcolor}
\usepackage{color}

\input xy
\xyoption{all} 

\numberwithin{equation}{section}

\newtheorem{theorem}{Theorem}[section]

\newtheorem{proposition}[theorem]{Proposition}

\theoremstyle{definition}

\newtheorem{remark}[theorem]{Remark}

\newcommand{\Id}{\mathbbmss{1}}

\newcommand{\rmi}{ \textnormal{i}}

\font\black=cmbx10 \font\sblack=cmbx7 \font\ssblack=cmbx5 \font\blackital=cmmib10  \skewchar\blackital='177
\font\sblackital=cmmib7 \skewchar\sblackital='177 \font\ssblackital=cmmib5 \skewchar\ssblackital='177
\font\sanss=cmss10 \font\ssanss=cmss8 
\font\sssanss=cmss8 scaled 600 \font\blackboard=msbm10 \font\sblackboard=msbm7 \font\ssblackboard=msbm5
\font\caligr=eusm10 \font\scaligr=eusm7 \font\sscaligr=eusm5  \font\fraktur=eufm10
\font\sfraktur=eufm7 \font\ssfraktur=eufm5 
\font\bsymb=cmsy10 scaled\magstep2
\def\all#1{\setbox0=\hbox{\lower1.5pt\hbox{\bsymb
       \char"38}}\setbox1=\hbox{$_{#1}$} \box0\lower2pt\box1\;}
\def\exi#1{\setbox0=\hbox{\lower1.5pt\hbox{\bsymb \char"39}}
       \setbox1=\hbox{$_{#1}$} \box0\lower2pt\box1\;}

\def\tx#1{{\fam0\relax#1}}

\newfam\bifam
\textfont\bifam=\blackital \scriptfont\bifam=\sblackital \scriptscriptfont\bifam=\ssblackital

\newfam\blfam
\textfont\blfam=\black \scriptfont\blfam=\sblack \scriptscriptfont\blfam=\ssblack

\newfam\bbfam
\textfont\bbfam=\blackboard \scriptfont\bbfam=\sblackboard \scriptscriptfont\bbfam=\ssblackboard

\newfam\ssfam
\textfont\ssfam=\sanss \scriptfont\ssfam=\ssanss \scriptscriptfont\ssfam=\sssanss

\newfam\clfam
\textfont\clfam=\caligr \scriptfont\clfam=\scaligr \scriptscriptfont\clfam=\sscaligr

\newfam\frfam
\textfont\frfam=\fraktur \scriptfont\frfam=\sfraktur \scriptscriptfont\frfam=\ssfraktur

\def\hpb#1{\setbox0=\hbox{${#1}$}
    \copy0 \kern-\wd0 \kern.2pt \box0}
\def\vpb#1{\setbox0=\hbox{${#1}$}
    \copy0 \kern-\wd0 \raise.08pt \box0}

\def\pmb#1{\setbox0\hbox{${#1}$} \copy0 \kern-\wd0 \kern.2pt \box0}
\def\pmbb#1{\setbox0\hbox{${#1}$} \copy0 \kern-\wd0
      \kern.2pt \copy0 \kern-\wd0 \kern.2pt \box0}
\def\pmbbb#1{\setbox0\hbox{${#1}$} \copy0 \kern-\wd0
      \kern.2pt \copy0 \kern-\wd0 \kern.2pt
    \copy0 \kern-\wd0 \kern.2pt \box0}
\def\pmxb#1{\setbox0\hbox{${#1}$} \copy0 \kern-\wd0
      \kern.2pt \copy0 \kern-\wd0 \kern.2pt
      \copy0 \kern-\wd0 \kern.2pt \copy0 \kern-\wd0 \kern.2pt \box0}
\def\pmxbb#1{\setbox0\hbox{${#1}$} \copy0 \kern-\wd0 \kern.2pt
      \copy0 \kern-\wd0 \kern.2pt
      \copy0 \kern-\wd0 \kern.2pt \copy0 \kern-\wd0 \kern.2pt
      \copy0 \kern-\wd0 \kern.2pt \box0}

\mathchardef\za="710B  
\mathchardef\zb="710C  
\mathchardef\zg="710D  
\mathchardef\zd="710E  
\mathchardef\zve="710F 
\mathchardef\zz="7110  
\mathchardef\zh="7111  
\mathchardef\zvy="7112 
\mathchardef\zi="7113  
\mathchardef\zk="7114  
\mathchardef\zl="7115  
\mathchardef\zm="7116  
\mathchardef\zn="7117  
\mathchardef\zx="7118  
\mathchardef\zp="7119  
\mathchardef\zr="711A  
\mathchardef\zs="711B  
\mathchardef\zt="711C  
\mathchardef\zu="711D  
\mathchardef\zvf="711E 
\mathchardef\zq="711F  
\mathchardef\zc="7120  
\mathchardef\zw="7121  
\mathchardef\ze="7122  
\mathchardef\zy="7123  
\mathchardef\zf="7124  
\mathchardef\zvr="7125 
\mathchardef\zvs="7126 
\mathchardef\zf="7127  
\mathchardef\zG="7000  
\mathchardef\zD="7001  
\mathchardef\zY="7002  
\mathchardef\zL="7003  
\mathchardef\zX="7004  
\mathchardef\zP="7005  
\mathchardef\zS="7006  
\mathchardef\zU="7007  
\mathchardef\zF="7008  
\mathchardef\zW="700A  
\mathchardef\zC="7009  

\newcommand{\be}{\begin{equation}}
\newcommand{\ee}{\end{equation}}

\newcommand{\bea}{\begin{eqnarray}}
\newcommand{\eea}{\end{eqnarray}}
\newcommand{\Z}{{\mathbb Z}}
\def\*{{\textstyle *}}
\newcommand{\R}{{\mathbb R}}

\newcommand{\s}{{\textstyle *}}

\def\xi{\tx{i}}

\def\s*{{\scriptstyle *}}

\newcommand{\beas}{\begin{eqnarray*}}
\newcommand{\eeas}{\end{eqnarray*}}

\author{Andrew James Bruce}
   \email{andrewjamesbruce@googlemail.com}
\begin{document}
\date{\today}
\title{A Novel Generalisation of Supersymmetry:  Quantum $\Z_2^2$-Oscillators and their `superisation'}

\begin{abstract}
We propose a very simple toy model of a $\Z_2^2$-supersymmetric quantum system and show, via Klein's construction, how to understand the system as being an $N=2$ supersymmetric system with an extra $\Z_2^2$-grading. That is, the commutation/anticommutation rules are defined via the standard boson/fermion rules, but the system still has an underlying $\Z_2^2$-grading that needs to be taken into account.
\end{abstract}

\maketitle

\section{Introduction}
Recently, a novel generalisation of supersymmetry that is inherently $\Z_2^2$-graded  ($\Z_2^2 := \Z_2 \times \Z_2$) has been proposed (see \cite{Bruce:2019}), and some classical and quantum models studied (see \cite{Aiwawa:2020a, Aiwawa:2023,Aiwawa:2020,Aiwawa:2020b,Bruce:2020a,Bruce:2021,Bruce:2020}). The basic degree of freedom in these models are a bosonic, exotic bosonic and two species of fermions. The novel aspect is that the exotic bosons anticommute with the fermions and the two species fermions mutually commute. In this sense these systems have exotic relative statistics, to use a term borrowed from  Green--Volkov parastatistics (see \cite{Druhl:1970} and references therein). Moreover, these systems have a pair of supersymmetry generators whose commutator, rather than anticommutator, is a (possibly vanishing for certain models) central term. At the time of writing, it is not clear if physical systems can exhibit this kind of generalised supersymmetry. Low dimensional systems are certainly a candidate as the spin-statistics theorem need not hold. For sure, more models need to be constructed, studied and their relation with conventional models made clear. There are promising results within multiparticle theories where, in principle, there are observable consequences of the $\Z_2^2$-grading (see \cite{Toppan:2020}). Furthermore, we remark that no experimental evidence today has emerged that nature realises supersymmetry at the level of fundamental particles. Rather speculatively, we consider the construction of simple models as an important step in experimentally realising $\Z_2^n$-supersymmetry in, say, trapped ion quantum simulators (see \cite{Cai:2022} for the case of a supersymmetric quantum mechanical model and \cite{HuertaAlderete:2017} for the case of para-boson oscillators).  We further remark that particles with exotic statistics have been proposed as a candidate for dark matter (see, for example, \cite{Nelson:2016}).  \par
With the above comments in mind, we construct a $\Z_2^2$-graded version of Nicolai's supersymmetric oscillator (see \cite{Nicolai:1976}), examine some of its elementary properties and then use Klein's construction (see \cite{Klien:1938}) to render the commutation/anticommutations to the the standard ones and the system supersymmetric.  The `superised' system is not completely standard as there is still an underlying $\Z_2^2$-grading:  the Hilbert--Fock space is $\Z_2^2$-graded and the supersymmetry generators still carry a $\Z_2^2$-grading.  The extra grading is encoded in two Witten parity operators rather than a single one as found in standard supersymmetry (this was first observed in \cite{Bruce:2020}). It is remarkable that a very simple mathematical system can exhibit this $\Z_2^2$-graded generalisation of supersymmetry.
\medskip 

\noindent \textbf{Conventions:} For notational simplicity, we work in units such the mass of the particles $m=1$ and Planck's constant $\hbar =1$ (as are any possible coupling constants). As an ordered set we define, $\Z_2^2 := \{(0,0), (1,1), (0,1), (1,0) \}$. The $\Z_2^2$-commutator is defined as
 $$[A,B]_{\Z_2^2} :=  AB -(-1)^{\langle \deg(A)| \deg(B)\rangle}\, BA\,,$$
where $\langle -|-\rangle$ is the standard scalar product. For  the conjugation operation we take the convention that $(ab)^\dag = b^\dag a^\dag$ irrespective of the $\Z_2^2$-degree. 

\section{Quantum $\Z_2^2$-Oscillators}
\subsection{The model via creation and annihilation operators}
We construct  a very simple a $\Z_2^2$-supersymmetric system over a single point, thus we have a model of zero-dimensional field theory.  This can be thought of as a spin-lattice model consisting  four independent spin degrees of freedom on a single lattice point. As the system is inherently zero-dimensional we do not have the proper notion of spin. None-the-less, we can consider operators that are $\Z_2^2$-graded and they satisfy canonical $\Z_2^2$-graded commutation  relations. That is, we take a $*$-algebraic approach. In particular,  consider a set of creation and annihilation operators acting a Hilbert space $\mathcal{H}$: \
\medskip 

\begin{tabular}{lll}
$b^\dag, b$ & degree $(0,0)$ & Standard Boson\\
$e^\dag, e$ & degree $(1,1)$ &  Exotic Boson\\
$f_1^\dag, f_1$ & degree $(0,1)$ &  Fermion Type 1\\
$f_2^\dag, f_2$ & degree $(1,0)$ &  Fermion Type 2\\
\end{tabular} 
\medskip

\noindent We will use classical notation $[-,-]$ and $\{ -,- \}$ for the $\Z_2^2$-commutators for clarity.  The commutation rules are the standard CCR and CAR,
\begin{align}
&[b, b^\dag] = 1, & [e, e^\dag] = 1, && \{f_i, f_i^\dag \}= 1,
\end{align}
with all other $\Z_2^2$-commutators vanish. For example, $f_1 f_2 = f_2 f_1$ and $ef_i =  -\,f_i e$.   We define the Hamiltonian as 
\begin{equation}\label{eqn:Ham1}
H_{00} := b^\dag b + e^\dag e + f_1^\dag f_1  + f_2^\dag f_2.
\end{equation}
This means that we are considering no interaction between the individual oscillators. Note that this Hamiltonian is the sum of the number operators $N_b = b^\dag b$, $N_e = e^\dag e$, $N_{f_1} = f^\dag_1 f_1$ and $N_{f_2} = f^\dag_2 f_2$. Thus, this system is the natural generalisation of Nicolai's supersymmetric oscillator (see \cite{Nicolai:1976}). The number operators are degree $(0,0)$ and satisfy the usual relations $[N_a, a^\dag] = a^\dag$ and $[N_a, a] = - a$, where $a \in \{ b, e, f_1, f_2\}$.  This system exhibits $\Z_2^2$-supersymmetry as first defined by Bruce \cite{Bruce:2019} and Bruce \& Duplij \cite{Bruce:2020}.  We naturally take observables to be degree $(0,0)$ self-adjoint operators. 
\begin{theorem}\label{thm:Z22QAlg}
The following self-adjoint $\Z_2^2$-graded operators
\begin{subequations}
\begin{align}
& Q_{01} := f^\dag_1 b + b^\dag f_1 + f^\dag_2 e + e^\dag f_2,\\
& Q_{10} := f^\dag_2 b + b^\dag f_2 + f^\dag_1 e + e^\dag f_1,
\end{align}
\end{subequations}
satisfy the $\Z_2 \times \Z_2$-graded, $\mathcal{N} = (1, 1)$ supertranslation algebra (with vanishing central term), i.e.,
\begin{align}
& \{ Q_{01}, Q_{01} \}=  \{ Q_{10}, Q_{10} \}= 2 H_{00}, &&  [Q_{10},Q_{01}] = 0.
\end{align}
\end{theorem}
\begin{proof}
This follows from a series of direct computations. First, using the commutation rule for $e$ and $e^\dag$
\begin{align*}
Q_{01}^2  &= f^\dag_1 b b^\dag f_1 + b^\dag f_1 f^\dag_1 b + f^\dag_2 e e^\dag f_2 + e^\dag f_2 f^\dag_2 e\\
& + (f^\dag_1 f^\dag_2 - f^\dag_2 f^\dag_1 )be +  (f_2 f^\dag_1 - f^\dag_1 f_2 ) b e^\dag + (f_1 f^\dag_2 - f^\dag_2 f_1 )b^\dag e +  (f_2 f_1 - f_1 f_2 ) b^\dag e^\dag,\\
& \textnormal{then using the commutation rules for } f_1, f_2 \textnormal{ and their conjugates the second line vanishes} \\
& = f^\dag_1 b b^\dag f_1 + b^\dag f_1 f^\dag_1 b + f^\dag_2 e e^\dag f_2 + e^\dag f_2 f^\dag_2 e,\\
&  \textnormal{next using the non-trivial commutation relations we obtain}\\
&= bb^\dag (\Id - f_1 f^\dag_1) + (\Id + b b^\dag)f_1 f^\dag_1
 + ee^\dag (\Id - f_2 f^\dag_2) + (\Id + e e^\dag)f_2 f^\dag_2\\
 &= b^\dag b + e^\dag e + f_1^\dag f_1  + f_2^\dag f_2 = H_{00}.
 \end{align*}
 The statement that $Q_{10}^2 = H_{00}$ follows from the above proof upon interchanging  $1 \leftrightarrow 2$ for the fermions.\par 
Moving on to the mixed commutator, writing out only the terms that do not obviously commute, we have 
\begin{align*}
[Q_{01},Q_{10}] & = f^\dag_2 bb^\dag f_1 - b^\dag f_1 f^\dag_2 b + f^\dag _2 b e^\dag f_2 - e^\dag f_2 f^\dag_2 b\\
&+ b^\dag f_2 f^\dag_1 b - f^\dag_1 b b^\dag f_2 + b^\dag f_2 f^\dag_2 e - f^\dag_2 e b^\dag f_2\\
&+ f^\dag_1 e b^\dag f_1 - b^\dag f_1 f^\dag_1 e + f^\dag_1 e e^\dag f_2 - e^\dag f_2 f^\dag_1 e\\
&+ e^\dag f_1 f^\dag_1  b - f^\dag_1 b e^\dag f_1 + e^\dag f_1 f^\dag_2 e - f^\dag_2 e e^\dag f_1,\\
& \textnormal{using the (anti)commutation rules we obtain} \\
& = f_1f_2^\dag - b e^\dag - f_2 f^\dag_1 + e b^\dag - eb^\dag + f_2 f^\dag_1 + b e^\dag - f_1 f^\dag_2 =0.
\end{align*}
The fact that $H_{00}$ is central, so $[H_{00}, Q_{01}] = [H_{00}, Q_{10}] =0$, follows from the Jacobi identity and the antisymmetry of the graded commutators. 
\end{proof}
\begin{remark} In general we have a central term $[Q_{10},Q_{01}] = 2 \rmi  \, Z_{11}$ where $Z_{11}$ is of $\Z_2^2$-degree $(1,1)$. Note that we have a commutator here and not an anticommutator as would be the case for standard supersymmetry. The vanishing of $Z_{11}$ is completely expected as all the oscillators are independent of each other.   Adding interactions require  the presence of coupling constants that carry a non-zero $\Z_2^2$-degree (see \cite{Aiwawa:2020b,Bruce:2020a, Bruce:2021}).
\end{remark}
\noindent \textbf{Observations:} \

\begin{enumerate}  
\item The  Hamiltonian \eqref{eqn:Ham1} is the sum of two bosonic and two fermionic harmonic oscillators ($\hbar = \omega = 1$). Thus, in the basis $| n_b, n_e , n_{f_1}, n_{f_2} \rangle $ the energy is given by $E = n_b +  n_e + n_{f_1}+ n_{f_2} = n$. See Proposition \ref{prop:DegSta} for the degeneracy of these states. The first four excited states are given in Table \ref{table:States}.
\item Clearly, $H_{00}|0,0,0,0 \rangle = 0$ and the zero energy ground state is a singlet. We use the shorthand $|0 \rangle := |0,0,0,0 \rangle $ for this ground state. The ground state being a zero energy state implies, as standard,  $Q_{01}|0 \rangle = 0$ and $Q_{10}|0 \rangle = 0$, meaning that $\Z_2^2$-supersymmetry is unbroken.  This is exactly the same situation as the standard supersymmetric oscillator.  The ground state is bosonic. 
\item There is a version of $R$-symmetry which shifts the $\Z_2^2$-degree and is given by 
\begin{align*}
b \mapsto \exp(\rmi \lambda)\,  e, & \quad e  \mapsto \exp(-\rmi \lambda)\, b, 
&\quad f_1 \mapsto \exp(\rmi \lambda)\, f_2, & \quad  f_2 \mapsto \exp(-\rmi \lambda)\, f_1,
\end{align*}
together with the conjugates, and here $\lambda \in \R$. Note that the Hamiltonian $H_{00}$ is invariant and that
 $$Q_{01} \longleftrightarrow Q_{10}\,.$$
\end{enumerate}
\begin{proposition}\label{prop:DegSta}
The $n^{th}$ energy level for $n \geq 1$ of the Hamiltonian \eqref{eqn:Ham1} is $4n$-fold degenerate.
\end{proposition}
\begin{proof}
For any fixed $n \geq 1$ and state is labelled $|n_b, n_e, n_{f_1}, n_{f_2} \rangle$ where $n_b, n_e \in \mathbb{N}$ and $n_{f_1}, n_{f_2} \in \{0,1 \}$, subject to the constraint  $n_b +  n_e + n_{f_1}+ n_{f_2} = n$. The fermionic labels are $00$, $01$, $10$ and $11$ and thus the proof reduces to arranging pairs of integers that sum to $n$, $n-1$ (counted twice) and $n-2$. This problem is equivalent to the number of ways one can place a single stick between $n$ balls, and then $n-1$ balls (counted twice) and finally $n-2$ balls. For the $n$-case we have 
\begin{align*}
| OO\cdots OO & = (0,n)\\
 O| O\cdots OO & = (1,n-1)\\
 & \vdots \\
  OO\cdots O|O & = (n-1,1)\\
   OO\cdots OO| & = (n,0)
\end{align*}
Thus, the number of pairs of numbers that sum to $n$ is $n+1$.  Then it is clear that the number of states for a given $n$ is just $n+1 + 2n + n-1 = 4n$.
\end{proof}
\renewcommand{\arraystretch}{1.5}
\begin{table}[h]
  \begin{tabular}{ | l | l | l | l |l|}
 \hline
  \multicolumn{1}{|c|}{Energy } &  \multicolumn{4}{|c|}{States }\\
    \hline
     & Boson & Exotic  &  Fermion 1 & Fermion 2 \\
       & $(0,0)$ & $(1,1)$   &  $(0,1)$  & $(1,0)$  \\ \hline
   0 & $|0,0,0,0 \rangle$ & - & - &-\\ \hline
    1 & $|1,0,0,0 \rangle$& $|0,1,0,0 \rangle$ &  $|0,0,1,0 \rangle$& $|0,0,0,1 \rangle$ \\
    \hline
    2 & $|2,0,0,0 \rangle$& $|1,1,0,0 \rangle$ &  $|1,0,1,0 \rangle$& $|1,0,0,1 \rangle$   \\
       & $|0,2,0,0 \rangle$& $|0,0,1,1 \rangle$ &  $|0,1,0,1 \rangle$& $|0,1,1,0 \rangle$ \\
    \hline
      3 & $|3,0,0,0 \rangle$& $|2,1,0,0 \rangle$ &  $|2,0,1,0 \rangle$& $|2,0,0,1 \rangle$   \\
       & $|1,2,0,0 \rangle$& $|1,0,1,1 \rangle$ &  $|1,1,0,1 \rangle$& $|1,1,1,0 \rangle$ \\
         & $|0,1,1,1 \rangle$& $|0,3,0,0 \rangle$ &  $|0,2,1,0 \rangle$& $|0,2,0,1 \rangle$ \\
         \hline
      4 & $|4,0,0,0 \rangle$& $|3,1,0,0 \rangle$ &  $|3,0,1,0 \rangle$& $|3,0,0,1 \rangle$   \\
       & $|2,2,0,0 \rangle$& $|2,0,1,1 \rangle$ &  $|2,1,0,1 \rangle$& $|2,1,1,0 \rangle$ \\
         & $|1,1,1,1 \rangle$& $|1,3,0,0 \rangle$ &  $|1,2,1,0 \rangle$& $|1,2,0,1 \rangle$ \\
         & $|0,4,0,0 \rangle$& $|0,2,1,1 \rangle$ &  $|0,3,0,1 \rangle$& $|0,3,1,0 \rangle$ \\
         \hline
   \end{tabular}
   \medskip 
   
   \caption{The first few energy levels of the $\Z_2^2$-oscillator \eqref{eqn:Ham1} in the `particle number' basis. }\label{table:States}
  \end{table}
We have a pair of Witten parity operators\footnote{also known as Klein operators, chirality operators or fermion number operators, though this last name is not appropriate in the current situation.} defined as
\begin{align}\label{eq:WitOPs}
& K_1 = \cos\big( \pi (N_e  + N_{f_1}) \big) , && K_2 = \cos\big( \pi (N_e  + N_{f_2}) \big),
\end{align}
which are both clearly $\Z_2^2$-degree $(0,0)$ and self-adjoint,  thus they correspond to observables. By construction we have that
\begin{equation}
K_i | n_b, n_e , n_{f_1}, n_{f_2} \rangle = (-1)^{n_e + n_{f_i}}| n_b, n_e , n_{f_1}, n_{f_2} \rangle.
\end{equation}
It is immediately clear that 
\begin{align}\label{eq:ComRulesW}
&[K_1, K_2] = 0, & [K_1, H_{00}] = [K_2, H_{00}] =0, && K_1^2 = K_2^2 = \Id.
\end{align}
In particular, the above implies that we can have simultaneous eigenfunctions of the Hamiltonian and the two Witten parity operators \eqref{eq:WitOPs}. We can then pick a basis for the states $| n, \varepsilon_1, \varepsilon_2 \rangle$, where $\varepsilon_i \in \{+1, -1\}$.  That is, the space of states $\mathcal{H}$ has a decomposition into four sectors depending on the sign of the Witten parity operators, i.e., $\mathcal{H} = \mathcal{H}_{++}\oplus\mathcal{H}_{--} \oplus\mathcal{H}_{+-} \oplus \mathcal{H}_{-+} $. These sectors correspond to bosons, exotic bosons, fermions of type 1 and fermions of type 2.  It is then convenient to relabel these sectors  via the corresponding $\Z_2^2$-degree, i.e.,   $\mathcal{H} = \mathcal{H}_{00}\oplus\mathcal{H}_{11} \oplus\mathcal{H}_{01} \oplus \mathcal{H}_{10}$.   Moreover, the Witten parity operators imply the superselection rule that only states that are homogeneous in $\Z_2^2$-degree are physically realisable (see, for example \cite{Giulini:2009}).  It is straightforward to observe that 
\begin{align}
& Q_{01}K_1 = - K_1 Q_{01}&&  Q_{10}K_1 = +K_1 Q_{10},\\
& Q_{01}K_2 = + K_1 Q_{01}&&  Q_{10}K_2 = - K_2 Q_{10},\nonumber 
\end{align}
and thus
\begin{align*}
& Q_{01}\mathcal{H}_{00} \subset \mathcal{H}_{01}, && Q_{10}\mathcal{H}_{00} \subset \mathcal{H}_{10}, \\
& Q_{01}\mathcal{H}_{11} \subset \mathcal{H}_{10}, && Q_{10}\mathcal{H}_{11} \subset \mathcal{H}_{01}, \\
& Q_{01}\mathcal{H}_{01} \subset \mathcal{H}_{00}, && Q_{10}\mathcal{H}_{01} \subset \mathcal{H}_{10}, \\
& Q_{01}\mathcal{H}_{10} \subset \mathcal{H}_{11}, && Q_{10}\mathcal{H}_{01} \subset \mathcal{H}_{11}. 
\end{align*}
The system really is $\Z_2^2$-supersymmetric, i.e., states from one sector are mapped to other sectors using $Q_{01}$ and $Q_{10}$.  It is important to note that applying $Q_{10}Q_{01}$ (or equivalent in this case $Q_{01} Q_{10}$) does \emph{not} return one to the starting sector as it would in standard supersymmetry.  For example, $Q_{10}Q_{01}\mathcal{H}_{00} \subset \mathcal{H}_{11}$.

\subsection{Klein operators and ``superisation''}
We now proceed to apply Klein's construction (see \cite{Klien:1938}) to redefine the operators we work with to render the system super, i.e.,  with the standard commutation/anticommutation rules for the creation and annihilation operators defined by a $\Z_2$-grading.  The natural choice here is to use the total degree of the assigned $\Z_2^2$-degree.  Moreover, we want the construction to lead to two standard supersymmetries.\par 
\begin{remark}
 Quesne  showed that the algebra of $\Z_2^2$-graded supersymmetric quantum mechanics is realisable in terms of a single bosonic degree of freedom using Calogero--Vasiliev algebras i.e., there is a minimal bosonisation of the theory (see \cite{Quesne:2021}). We will content ourselves with a `superisation' in this note. 
\end{remark}
We have to chose one of the Witten operators \eqref{eq:WitOPs} to be our Klein operator. We pick $K_1$ for no particular reason other than our choice of ordering with the elements of $\Z_2^2$. We then define a new set of fermionic creation and annihilation operators as 
\begin{subequations}
\begin{align}
& a_1 = f_1 K_1, && a^\dag_1 = K_1 f^\dag_1,\\
& a_2 = f_2 K_1, && a^\dag_2 = K_1 f^\dag_2.
\end{align}
\end{subequations}
\begin{proposition}
The set of operators $\{ b, b^\dag, e , e^\dag, a_1, a^\dag_1, a_2, a^\dag_2 \}$ satisfy the standard commutation/anticommutation rules  for a pair of bosonic and a pair of fermionic creation and annihilation operators where the supercommutation rules are defined by the total degree of the operators. 
\end{proposition}
\begin{proof}
This follows from the properties of the Witten parity operators. We do not need to check all expressions, just those that involve $a_1$ and $a_2$ (and their conjugates). Moreover, we need not consider any expression involving $b $ (and its conjugate).  For example, $\{ e, f_1\} K_1= ef_1K_1 + f_1e K_1 =ef_1  K_1 - f_1 K_1 e = [e, a_1] =0$. Similarly, $K_1[f_2, f_1]K_1  = K_1 f_2 f_1 K_1 - K_1 f_1 f_2 K_1 = f_2 f_1 K_1 +  f_1 K_1f_2 K_1 = \{a_2,a_1 \} =0$. Finally, just to further illustrate the point, $K_1\{ f_1, f^\dag_1\}K_1 = K_1 f_1 f^\dag_1 K_1 + K_1 f^\dag_1 f_1 K_1 =  f_1K_1 K_1 f^\dag_1 + K_1 f^\dag_1 f_1 K_1  =  \{a_1, a^\dag_1 \} =1$.  All other commutators and anticommutators can similarly be deduced.  The claim that these are now all supercommutators using the total degree follows directly.  
\end{proof}
Furthermore we define the following self-adjoint operators:
\begin{subequations}
\begin{align}\label{eqn:Ham}
& H  := b^\dag b + e^\dag e + a^\dag_1 a_1 + a^\dag_2 a_
2 =  H_{00},\\ \label{eqn:Q1}
& Q_{1} := \rmi K_1 Q_{01} = \rmi a^\dag_1 b - \rmi b^\dag a_1 + \rmi a^\dag_2 e - \rmi e^\dag a_2,\\\label{eqn:Q2}
& Q_{2} := K_1 Q_{10} = a^\dag_2 b + b^\dag a_2 + a^\dag_1 e  + e^\dag a_1.
\end{align}
\end{subequations}
Note that the Hamiltonian is unchanged, but now has the interpretation of the sum of the Hamiltonians for a pair of distinguishable bosons and a pair of distinguishable fermions. 
\begin{theorem}
The above operators \eqref{eqn:Ham}, \eqref{eqn:Q1} and  \eqref{eqn:Q2} satisfy the $\mathcal{N}=2$ super-translation algebra (with vanishing central charge)
\begin{align*}
& \{ Q_1 , Q_1 \} = \{Q_2, Q_2\} = 2 H, && \{ Q_2, Q_1\} =0,
\end{align*}
and all other commutators vanishing.
\end{theorem}
\begin{proof}
Note that we are now dealing with commutators/anticommutators defined by the total degree of the operators. Direct computation using the properties of the Witten operators and Theorem \ref{thm:Z22QAlg}.
\begin{align*}
& \{ Q_1, Q_1\} = 2 \rmi K_1 Q_{01}\rmi K_1 Q_{01} = 2 Q_{01}Q_{01} = 2 H_{00} = 2 H,\\
& \{ Q_2, Q_2\} = 2 K_1 Q_{10} K_1 Q_{10} = 2 Q_{10}Q_{10} = 2 H_{00} = 2 H, \\
& \{ Q_2, Q_1 \} = \rmi ( K_1 Q_{10}K_1Q_{01} + K_1 Q_{01}K_1 Q_{10}) = \rmi (Q_{10} Q_{01} - Q_{01}Q_{10}) =0. 
\end{align*}
Checking that the Hamiltonian is central is similarly straightforward,
\begin{align*}
[H, Q_1] = \rmi K_1 [H_{00}, Q_{01}] =0 , && [H, Q_2] =  K_1 [H_{00}, Q_{10}] =0.
\end{align*}
\end{proof}
\begin{remark}
This construction is not completely canonical, there is the other obvious choice of using $K_2$ and the obvious amendments to the above constructions. This would be no more than exchanging the labelling of fermions of type 1 and type 2. 
\end{remark}
The Lie superalgebra formed by $H, Q_1$ and $Q_2$ should be considered as a $\Z_2^2$-graded Lie superalgebra, i.e., a Lie superalgebra with an additional compatible $\Z_2^2$-grading. Specifically, $H$ is even and carries $\Z_2^2$-degree $(0,0)$,  $Q_1$ is odd and carried $\Z_2^2$-degree $(0,1)$, and $Q_2$ is odd and carried $\Z_2^2$-degree $(1,0)$. Of course, these operators still act on the Hilbert-Fock space  $\mathcal{H} = \mathcal{H}_{00}\oplus\mathcal{H}_{11} \oplus\mathcal{H}_{01} \oplus \mathcal{H}_{10}$, and the $\Z_2^2$-grading still needs to be taken into account. The Witten  parity operators encode the $\Z_2^2$-grading and one can easily deduce the following: 
\begin{align}\label{eqn:KQComms}
& \{K_1,  Q_{1}\} = 0, &&  [K_1,  Q_{2}] =0,\\
&  [K_2,  Q_{1}] =0,&&  \{K_2,  Q_{2}\}=0.\nonumber 
\end{align}
Just as before, we have
 \begin{align*}
& Q_{1}\mathcal{H}_{00} \subset \mathcal{H}_{01}, && Q_{2}\mathcal{H}_{00} \subset \mathcal{H}_{10}, \\
& Q_{1}\mathcal{H}_{11} \subset \mathcal{H}_{10}, && Q_{2}\mathcal{H}_{11} \subset \mathcal{H}_{01}, \\
& Q_{1}\mathcal{H}_{01} \subset \mathcal{H}_{00}, && Q_{2}\mathcal{H}_{01} \subset \mathcal{H}_{10}, \\
& Q_{1}\mathcal{H}_{10} \subset \mathcal{H}_{11}, && Q_{2}\mathcal{H}_{01} \subset \mathcal{H}_{11}. 
\end{align*}
Via this construction, the relative statistics of the creation/annihilation operators are now standard.  Moreover, the system exhibits  supersymmetry, but now an extra internal  quantum number - the $\Z_2^2$-grading that is encoded in the two Witten parity operators.  These observations sit comfortably with the results of \cite{Druhl:1970}. In particular, systems with para-fermions can, under some technical conditions, be reformulated to have standard statistics, but now the observable algebra is selected by a non-Abelian gauge group.  This supersymmetric system is not entirely standard. Supersymmetry generators usually anticommute with the Witten parity operator, but in the current situation we have both commutators and anticommutators, see \eqref{eqn:KQComms}. Generalising supersymmetry to include internal degrees of freedom - and we view the extra $\Z_2^2$-grading in this light -  has a long history dating back to the late 1970s (see \cite{Fayet:1977} and references therein).

\section{Concluding remarks}
In this short note, we have constructed a simple $\Z_2^2$-supersymmetric model based on creation and annihilation operators. We have shown that via Klein's construction can be rendered supersymmetric. \par 
One issue here is that the model does not have a central charge and it is desirable to amend this. The lack of central charge is due to the four oscillators not interacting with each other.  It has already be noticed that interacting models, classical at least, seem to require coupling constants that carry non-zero $\Z_2^2$-grading. The physical interpretation of such constants is not exactly clear, nor is the role of exotic bosons in nature.\par  
 Moreover, building simple models with $\Z_2^n$-supersymmetry for $n>2$ is a challenge as the number of elements of $\Z_2^n$ grows exponentially as $n$ increases.  This increase in the number of degrees of freedom has hindered model-building (for work in this direction see \cite{Aiwawa:2020a}).  
 
 \section*{Acknowledgements}
 The author thanks Steven Duplij  and Francesco Toppan for their friendship and discussions. Cordial thanks are extended to the referee for their careful reading of the manuscript and helpful comments.

\end{document}